\documentclass[reqno]{amsart}
\usepackage[T1]{fontenc}
\usepackage[utf8]{inputenc}
\usepackage{amsmath,amssymb,amsthm,mathtools}
\usepackage{lmodern,xspace}
\usepackage[a4paper,margin=3cm]{geometry}
\usepackage{hyperref}
\hypersetup{colorlinks=true,linkcolor=blue,citecolor=blue,urlcolor=blue}
\usepackage{microtype}

\newtheorem{theorem}{Theorem}[section]
\newtheorem{lemma}[theorem]{Lemma}
\newtheorem{corollary}[theorem]{Corollary}
\newtheorem{proposition}[theorem]{Proposition}
\theoremstyle{definition}
\newtheorem{definition}[theorem]{Definition}
\theoremstyle{remark}
\newtheorem{remark}[theorem]{Remark}

\newcommand{\bthm}{\begin{theorem}}\newcommand{\ethm}{\end{theorem}}
\newcommand{\blem}{\begin{lemma}}\newcommand{\elem}{\end{lemma}}
\newcommand{\bcor}{\begin{corollary}}\newcommand{\ecor}{\end{corollary}}
\newcommand{\bpr}{\begin{proposition}}\newcommand{\epr}{\end{proposition}}
\newcommand{\bdf}{\begin{definition}}\newcommand{\edf}{\end{definition}}
\newcommand{\brm}{\begin{remark}}\newcommand{\erm}{\end{remark}}
\usepackage[utf8]{inputenc}
\usepackage[T1]{fontenc}
\usepackage{amsmath,amsthm,amssymb,amsfonts,physics}
\usepackage{hyperref}
\usepackage{mathrsfs}
\usepackage{graphicx}
\usepackage{bm}
\usepackage{enumitem}

\newcommand{\de}{{\rm d}}
\newcommand{\R}{\ensuremath{\mathbb R}\xspace}

\theoremstyle{definition}
\theoremstyle{remark}
\newcommand{\A}{\mathcal A}

\newcommand{\bt}{\begin{theorem}}                     
	\newcommand{\et}{\end{theorem}}                       
\newcommand{\bd}{\begin{definition}}                  
	\newcommand{\ed}{\end{definition}}                    
\newcommand{\bl}{\begin{lemma}}                       
	\newcommand{\el}{\end{lemma}}                                   
\newcommand{\bere}{\begin{remark}}                      
	\newcommand{\ere}{\end{remark}}                       

\newcommand{\beq}{\begin{equation}}
	\newcommand{\eeq}{\end{equation}}
\def\bal#1\eal{\begin{align}#1\end{align}}              
\def\baln#1\ealn{\begin{align*}#1\end{align*}}          
\def\bml#1\eml{\begin{multline}#1\end{multline}}        
\def\bmln#1\emln{\begin{multline*}#1\end{multline*}}  
\def\bga#1\ega{\begin{gather}#1\end{gather}}
\def\bgan#1\egan{\begin{gather*}#1\end{gather*}}

\title[Static spacetimes  with a Finsler angular sector]{Static spacetimes  with  a Finsler angular sector}
\author[E. Caponio]{Erasmo Caponio}
\address{Dipartimento di Meccanica, Matematica e Management, Politecnico di Bari, Italy}
\email{erasmo.caponio@poliba.it}

\date{}
\begin{document}
\begin{abstract}
We consider  static spacetimes  in spherical coordinates whose angular sector is described by a Finsler metric rather than the standard round metric on $S^2$.  
Our first contribution is \emph{kinematical}: maintaining arbitrary  lapse and radial factors $e^{\nu(r)}$, $e^{\vartheta(r)}$, and relying solely on Killing symmetries and the null constraint, we derive model--independent relations for circular photon orbits and the effective  dynamics. 
By specializing the angular sector to  Randers sphere of constant  positive  flag curvature,  we obtain  exact expressions for the conserved angular charge, the critical impact parameter and we quantify a Finslerian Sagnac--type effect. 

\noindent Our second contribution is \emph{dynamical}: we examine  the field equations used in the literature to determine  $(e^{\nu},e^{\vartheta})$. We revisit the family of hairy black holes  in \cite{Nekouee2025}, 
 demonstrating that  the analysis therein neglects crucial non-reversible Finsler features. Furthermore, we show  that the  solutions presented as new reproduce previously known results in \cite{Ovalle2021}.

\end{abstract}
	\maketitle

\section{Introduction}
The quest to understand  gravity beyond Einstein’s general relativity motivates geometric extensions of spacetime. Among them, Lorentz--Finsler geometry relaxes the quadratic restriction on the metric by replacing it with a positively homogeneous, of degree $2$ Lagrangian, while retaining key structures such as causal cones and Noether symmetries. Because the metric depends on both position and direction, anisotropy becomes a built-in feature. This provides a natural framework to model departures from local isotropy or  Lorentz invariance within a well-defined causal theory \cite{perlick06,Minguzzi2015CMP,Minguzzi2019RMP,JavaloyesSanchez2020RACSAM,HohmannPfeiferVoicu2022}.

An  interesting class of Lorentz--Finsler spacetimes arises when the Finslerian structure is confined to specific sectors of the geometry. In this work, we consider static spacetimes in spherical coordinates where the sphere $S^2$ (also called here {\em angular sector}) carries a positive, strongly convex,  Finsler metric \cite{Bao2000} while the temporal and radial sectors remain Lorentzian. 
This hybrid structure provides a tractable framework for exploring  Finslerian effects while maintaining a structure familiar from general relativity \cite{CapSta16}. 
The study of black hole solutions in this class initiated with the work of Li and Chang \cite{LiChang2014}. These solutions naturally raise questions about how the direction--dependent metric structure affects fundamental black hole properties such as photon spheres, shadows, and thermodynamic characteristics. 

The choice of the Finsler metric on the angular sector is not arbitrary, but the mathematical freedom available in the Finslerian setting is remarkably different from the Riemannian case. In Riemannian geometry, the Killing--Hopf theorem severely constrains the possible metrics of constant curvature on a two-sphere: up to diffeomorphisms  and an overall scale factor, there is  only the round metric. In  the Finslerian setting,  Bryant and Shen \cite{Bryant1996,Shen2002} demonstrated that there exist infinitely many non-isometric Finsler metrics on $S^2$ with the same constant positive flag curvature. This remarkable non-uniqueness opens up a rich landscape of possibilities absent in the Riemannian case and  suggests that Finslerian angular sectors could encode additional physical degrees of freedom not present in general relativity. 
The choice of metrics with constant flag curvature is  strongly motivated  when one considers a Finslerian generalization of the Einstein's field equations  proposed by Rutz \cite{Rutz} (originally, just for vacuum solutions) and based on    the Finslerian Ricci scalar. In fact, the constant flag curvature condition gives  that  the Finslerian Ricci scalar of the angular metric is constant as well,  simplifying the field equation. 

Among metrics with constant flag curvature, Randers metrics --which can be understood as Riemannian metrics perturbed by a one-form-- are particularly appealing due to their connection with Zermelo's navigation problems and their relatively simple geodesic structure \cite{Robles2007}. Within this family  one finds both metrics with all geodesics closed  and   metrics admitting just two closed geodesics (the so-called {\em Katok} examples).
 
Our work has two main objectives. First, we provide a mathematically rigorous treatment of  null geodesics in static Lorentz--Finsler spacetimes with a Finsler angular sector (focusing in particular on Randers ones), carefully considering  how the anisotropic structure affects the conserved quantities. We show that while the condition for circular photon orbits (the photon sphere) remains formally identical to the Lorentzian case, the critical impact parameter acquires direction--dependent corrections that encode the Finslerian nature of the angular geometry. A  Sagnac--type effect, where co-rotating and counter-rotating photons on the photon sphere have different periods as measured by the observer $\partial_t$ is also described.

Second, we discuss some critical aspects of \cite{Nekouee2025}; as shown in Remark~\ref{comparison}, the {\em equatorial restriction} is not “without loss of generality’’ even  for Randers metrics. Moreover,   some results in \cite{Nekouee2025} overlook orientation effects due to non-reversibility. More important,  we show that the claimed new solutions in \cite{Nekouee2025} are  essentially identical to the hairy black hole solutions previously obtained by Ovalle et al. \cite{Ovalle2021} using gravitational decoupling in the Lorentzian setting.

The paper is organized as follows. In Section~\ref{LFR}, we establish the Lorentz--Finsler framework and introduce the specific form of static Finsler spacetimes with angular Randers  structure having constant flag curvature. Section~\ref{circular} examines circular geodesics and derives the generalized photon sphere condition; moreover we show how the Randers structure induces a Finslerian Sagnac effect. In Section~\ref{appendix}, we show that it is possible  to discriminate between a stationary Lorentzian and a static Lorentz--Finsler metric under shared  optical data, i.e. when the two metrics  provide the same Sagnac--type delay, by evaluating the proper time of a worldline in the two  geometries.   Section~\ref{FFE} discusses the Finslerian field equations and their simplification for constant flag curvature angular metrics. In particular, Subsection~\ref{nekouee} provides a critical analysis of  \cite{Nekouee2025} establishing proper attribution of results contained there. Finally, Section~\ref{end} presents our conclusions. 
\section{Lorentz--Finsler Randers framework}\label{LFR}
In a Finsler spacetime, the fundamental geometric object is not a metric tensor but rather a function $L: TM \to \mathbb{R}$ defined on the tangent bundle $TM$ of  the four dimensional manifold $M$. This function, which we call the {\em Lorentz--Finsler structure} or {\em metric}, must satisfy certain conditions to provide a coherent geometric framework.  It must be positively homogeneous of degree two in the fiber coordinates, meaning that $L(x, \lambda y) = \lambda^2 L(x, y)$ for any positive scalar $\lambda$, and for all $x\in M$, $y\in T_x M$. This property ensures that the notion of length for a causal curve, i.e. a curve $\gamma\colon I \to M$, $\gamma=\gamma(s)$ satisfying  $L(\gamma(s),\dot\gamma(s))\geq 0$, remains well-defined under  orientation preserving reparametrization, where 
the length of $\gamma$ is defined as $\int_I\sqrt{L(\gamma, \dot\gamma)}ds.$ 

The {\em fundamental tensor}, which plays a role analogous to the metric tensor in Riemannian geometry, is defined   through the relation 
\begin{equation}\label{fundtens}
		g_{\mu\nu}(x,y) = \frac{1}{2}\frac{\partial^2 L}{\partial y^\mu \partial y^\nu}(x,y),\quad \mu,\nu\in\{0,\ldots,3\}, \  \text{for all } (x,y)\in \mathcal A,
\end{equation} 
			where $\A\subset TM\setminus 0$ is an open subset where $L$ is smooth and such that $\pi(\A)=M$ ($\pi:TM\to M$ being the canonical projection) and, for each $v\in\A$, $\lambda v \in \A$ for all $\lambda>0$.			
			The crucial difference from Riemannian geometry is that $g$   depends not only on the position $x$ but also on the direction $y$, i.e. $g_{\mu\nu}$ is a section of the pullback bundle
			$\pi^*(TM\otimes TM)\to\mathcal A$, not a tensor on $M$.

We require $g(x,y)$ to be of Lorentzian signature type for all $(x,y)\in \A$.\footnote{We adopt signature $(+,-,-,-)$;  moreover, units with $G=c=1$ are used.}

The rich mathematical structure of Finsler geometry for {\em classical Finsler metrics}, i.e. when $L$ is positive and $g$ positive definite for any $(x,y)\in TM\setminus 0$, is presented comprehensively in the influential  text by Bao, Chern, and Shen \cite{Bao2000}. 

In this work,  we consider static Finsler spacetimes where the Finsler structure takes the form \begin{equation}\label{ansatz}
	L(t,r,\theta,\phi, y^t , y^r, y^\theta, y^\phi) = e^{\nu(r)}(y^t)^2  - e^{\vartheta(r)}(y^r)^2 - r^2\overline{F}^2(\theta,\phi,y^\theta,y^\phi),
	\end{equation} 
	where $y=(y^t,y^r,y^\theta,y^\phi)\in T_{x}M$, $x=(t,r,\theta,\phi)$. 
	This ansatz formally resembles a classical static spherically   symmetric configuration but includes  non-trivial Finslerian  effects through the {\em angular sector metric} $\overline{F}$ which  is a classical Finsler metric on $S^2$, i.e. it is non-negative, positively homogeneous of degree $1$ and its fundamental tensor (associated to $L=\overline F^2$) is positive definite for all $\theta,\phi,y^\theta,y^\phi$ in the slit tangent bundle $TS^2\setminus 0$.    

\bere\label{regularity}
For the static Lorentz--Finsler metric $L$ in \eqref{ansatz},
$\overline F^2$ is  $C^1$ on $TM$ but it is $C^2$ (or more regular) only on $TS^2\setminus\{0\}$. Accordingly, the fundamental tensor $g$ in \eqref{fundtens}
is well defined on the \emph{admissible set}
\[
\mathcal A\;:=\;\big\{(x,y)\in TM\setminus\{0\}\ :\ (y^\theta,y^\phi)\neq(0,0)\big\}.
\]
Along  directions $(y^t,y^r,0,0)$ the $(\theta,\phi)$–block of $g_{\alpha\beta}$ is not defined (the  $(t,r)$–block is). Whenever  some ``tensorial object'',  constructed by contraction with $g_{\mu\nu}(x,y)$ or/and  its inverse $g^{\mu\nu}(x,y)$, is  required at such directions, we agree to define it as a limit along direction in $\mathcal A$  (provided that the limit exists).
\ere
Metrics of the type \eqref{ansatz} were considered in \cite{LiChang2014};   they belong
to the larger class of static Lorentz--Finsler metrics  that  were first introduced  by \cite{LPH} and subsequently   studied in \cite{CapSta16} (see also \cite{CapSta18}). 
\subsection{Randers angular sector}
As we will see in Section~\ref{FFE}, a physically reasonable choice  of the angular metric $\overline{F}$ is to take one  with constant flag curvature (see \cite[\S 3.9]{Bao2000} for the definition of the flag curvature).  Unlike the Riemannian case,  there exist infinitely many non-isometric Finsler metrics on $S^2$ with the same constant positive flag curvature, as shown by Bryant and Shen \cite{Bryant1996, Shen2002}. In particular,  the family of Randers metrics with constant flag curvature $1$ in \cite{Shen2002} is constituted by non-locally projectively flat metrics.

 A Randers metric is a positive definite Finsler metric of the form $\overline{F} = \alpha + \beta$, where $\alpha$ is the norm of a Riemannian metric $a$ and  $\beta$ is a one-form with norm w.r.t. $a$ everywhere strictly less than $1$.
More explicitly, we have 
\begin{equation}\label{Randers}
			\begin{aligned}\alpha(\theta,\phi;y^\theta,y^\phi) &= \sqrt{a_{AB}(\theta,\phi)y^A y^B}, \\
			\beta(\theta,\phi;y^\theta,y^\phi) &= b_A(\theta,\phi)y^A,\qquad a^{AB}b_A b_B<1,
		\end{aligned}
		\end{equation}
where the indices $A,B$ run over the angular coordinates $\{\theta,\phi\}$. The physical interpretation of a Randers metric with constant positive flag curvature  becomes clearer when we consider   Zermelo navigation, where we imagine navigating on the round sphere in the presence of a ``wind'' field $W$ (see \cite{BaoRoblesShen2004, CJS14}). At each point on $S^2$, the wind modifies the indicatrix of the round metric by translating it; as a consequence the  distance function is modified as well in a direction--dependent way.

Let  $K>0$ and let us consider on $S^2$:
\[
h_K = \frac{1}{K}\big(d\theta^2 + \sin^2\theta \, d\phi^2\big), 
\qquad 
W = \varepsilon\,\partial_\phi, \qquad 0\leq \varepsilon<\sqrt{K}.
\]
Then $\|W\|_{h_K}^2 = \tfrac{\varepsilon^2}{K}\sin^2\theta$. Let us set
\beq\label{lambdak}
\lambda_K(\theta): = 1 - \frac{\varepsilon^2}{K}\sin^2\theta.
\eeq

For a tangent vector $y=y^\theta\partial_\theta+y^\phi\partial_\phi$, the 
Randers metric obtained by Zermelo navigation with data $h_K$ and $W$ is (see \cite[\S 1.1.2]{BaoRoblesShen2004}):
\beq\label{Randersphere}
F_K(\theta,\phi;y^\theta,y^\phi)
=\dfrac{1}{\sqrt K \lambda_K(\theta)}\sqrt{\,\lambda_K(\theta)\,(y^\theta)^2
		+\sin^2\theta\,(y^\phi)^2\,}
	-\dfrac{\varepsilon \sin^2\theta}{K\lambda_K(\theta)}y^\phi.
\eeq
Thus, as a particular case of \cite[Th. 5.1]{BaoRoblesShen2004} (see also \cite[Rmk. 3.1]{Shen2002})  we have:
\begin{proposition}
	Let $\overline F$  be a Randers metric on $S^2$ of constant positive flag curvature $K$.  
	Then, locally up to isometry, there exist a Killing vector field $X$ of the round metric on $S^2$ and a constant $\varepsilon\in [0,\sqrt K)$, such that in spherical coordinates $(\theta,\phi)$ adapted to the 
	$SO(2)$-action generated by $X$ (i.e. the fixed points lie at $\theta=0,\pi$ and   $X=\partial_\phi$),
	$\overline F$ can be written as in \eqref{Randersphere} with the wind $W=\varepsilon X$.
	
	Conversely, for any choice of $K>0$ and $\varepsilon\in\R$ with  $|\varepsilon|<\sqrt K$,  	 \eqref{Randersphere} defines  a Randers metric on $S^2$ of constant flag curvature $K$. 
\end{proposition}
\begin{remark}[Symmetries of Randers spheres with constant flag curvature]
	\label{rem:axial}
	According to the classification in \cite{BaoRoblesShen2004}, any strongly convex Randers metric $\overline F$ on $S^2$ with constant positive flag curvature $K$ is obtained (up to local isometry) via Zermelo navigation on the round sphere $(S^2,h_K)$ under the influence of a Killing wind $W$ satisfying $\|W\|_{h_K}<1$.
	
	The symmetries of the Randers metric are precisely the background isometries that leave the wind invariant. Specifically, from \cite[Lemma 1.2]{BaoRoblesShen2004}, the isometry group and its Lie algebra are given by:
	\begin{equation}\label{isomkill}
	\mathrm{Isom}_0(\overline F)=\{\varphi\in\mathrm{Isom}(h_K):\ \varphi_*W=W\},
	\qquad
	\mathfrak{kill}(\overline F)=\{X\in\mathfrak{so}(3):\ [X,W]=0\}.
	\end{equation}
	We can, without loss of generality, align the spherical coordinate system so that $W$ coincides with the azimuthal generator (i.e., we rotate the frame such that the axis of rotation generated by $W$ coincides with the standard $z$-axis). 	In this adapted frame, \eqref{isomkill}  reduce to $\mathrm{Isom}_0(\overline F)\cong \mathrm{SO}(2)$ and $\mathfrak{kill}(\overline F)=\operatorname{span}\{W\}$.
	Besides this,  the  equatorial reflection $(\theta,\phi)\mapsto(\pi-\theta,\phi)$  also preserves $W$; hence the  full isometry group is  $\mathrm{Isom}(\overline F)=\mathrm{SO}(2)\times\mathbb Z_2$. 
	
	Consequently, the  Lorentz--Finsler metric in \eqref{Randersphere} is static and axially symmetric. It recovers the full spherical symmetry  only in the limit of vanishing wind, $\varepsilon \to 0$.
	\end{remark}

\section{Circular geodesics, critical impact parameter and a  Sagnac--type effect}\label{circular}
Any geodesic $\gamma=\gamma(s)$ of $L$, i.e. any critical point   of the functional $\gamma: I\subset \R\to M\mapsto \int_IL(\gamma,\dot\gamma)\de s$, defined on the path space of sufficiently regular curves connecting two fixed points on $M$, satisfies  the conservation law:
\begin{equation*}
	L(\gamma,\dot\gamma) = \kappa
\end{equation*}
where $\kappa$ is a constant which is, by definition, $0$ for {\em null} geodesics and positive  (resp. negative)
for {\em timelike}  (resp. {\em spacelike}) ones  (see  \cite[Th. 2.13]{CapSta16}).

Let us denote by $E$ the Noether charge associated with the Killing vector field $\partial_t$ of the Lagrangian $\frac 12 L$ in \eqref{ansatz}:
\begin{equation}\label{E}
	E(x,y):=\partial_y (\tfrac 12 L)(x,y)[\partial_t]=e^{\nu}y^t .
\end{equation}
It is well known that  along  a geodesic $\gamma=\gamma(s)$, $E\big(\gamma(s), \dot\gamma(s)\big)$ is also constant (see, e.g., \cite[Rmk. 2.4]{capcor}); such a constant will be simply denoted   with $E$.
We note that for {\em causal geodesics} (i.e. the ones with $\kappa\geq 0$) $E\neq 0$ and therefore we can divide them in {\em future--pointing} and {\em past--pointing} ones according to $E>0$ or $E<0$ respectively.

The Euler–Lagrange equation
$\frac{d}{ds}(\partial_{y^r}L)-\partial_r L=0$
and the assumption $\partial_r\overline F^2=0$ yield for a geodesic  with $r$--component constant 
\beq\label{rconst}
\dot r=0,\qquad \frac12\,e^{\nu(r)}\,\nu'(r)\,\dot t^2=r\,\overline F^2(\theta,\phi,\dot \theta,\dot \phi).
\eeq
\subsection{Photon sphere}
We now derive the  condition that null circular orbits must satisfy without restricting the analysis to equatorial geodesics $\theta=\frac{\pi}{2}$.
\begin{proposition}\label{circulargeo}
Let $(M,L)$ be a Finsler spacetime with $L$ as in \eqref{ansatz}.  If a null geodesic has constant $r$--component $r(s)\equiv r_C$ (a circular photon orbit), then $r_C$ satisfies
\begin{equation}\label{photonsphere}
r_C\,\nu'(r_C)=2.
\end{equation}
Moreover, along such an orbit one has the constraint
\begin{equation}\label{null}
E^2 e^{-\nu(r_C)} \;=\; r_C^2\,\overline F^2\big(\theta(s),\phi(s), \dot\theta(s),\dot \phi(s)\big)\,.
\end{equation}
Conversely, if $r_C$ obeys \eqref{photonsphere} and there exists a non-constant $\overline F$--geodesic on $(S^2,\overline F)$ with constant $\overline F$--speed such that \eqref{null} holds for some $E> 0$, then its lifts with $r\equiv r_C$ and $\dot t=\pm E e^{-\nu(r_C)}$ are, respectively, a future--pointing and a past--pointing null circular geodesic.
\end{proposition}
\begin{proof}
	Recalling that the Noether charge \eqref{E} is conserved along geodesics,  the time component of the velocity satisfies  \begin{equation}\label{dott}
		\dot{t} = E e^{-\nu(r)}.
	\end{equation}
	For a circular null geodesic (where $L=0$, $\dot{r}=0$, and $r=r_C$), substituting this expression for $\dot{t}$ into the null constraint $e^{\nu(r)}\dot{t}^2 - r^2 \overline{F}^2 = 0$ yields:
	\begin{equation*}
		E^{2}e^{-\nu(r_C)}=r_C^{2}\overline{F}^{2}.
	\end{equation*}
	Substituting this into \eqref{rconst} gives $r_C \nu'(r_C)=2$. This shows the radius is purely determined by the Lorentzian sector associated with coordinates $(t,r)$ and not by $\overline{F}$. With $r \equiv r_C$, the remaining equations in $\theta$ and $\phi$ in the Euler-Lagrange system reduce to the geodesic equations of the Finsler metric $\overline{F}$ on $S^2$ (with constant $\overline{F}$-speed fixed by (10)). The converse follows analogously as $L$ depends on $(\theta,\phi,\dot{\theta},\dot{\phi})$ only through the Finsler metric $\overline{F}$ and \eqref{dott} implies that $\dot t$ is constant if $r$ is constant.
\end{proof}
\begin{remark}
Equation \eqref{photonsphere} is the same as in the spherical symmetric Lorentzian setting. The role of $\overline F$ is twofold: (i) it fixes the  curves on $S^2$ (the $\overline F$--geodesics), and (ii) it enters the \emph{impact parameter} via \eqref{null} when an additional axial symmetry is present (see \S~\ref{subsec:axial-charge}).

Let us see some explicit examples:
\begin{itemize}
	\item \textit{Schwarzschild-like metric:} for the classic case $e^{\nu(r)} = 1 - \frac{2M}{r}$, we have $\nu'(r) = \frac{2M}{r(r-2M)}$. The condition becomes $\frac{2M}{r_C-2M} = 2$, which yields the unique physically relevant solution $r_C = 3M$.
	\item \textit{Power-law metric:} consider the case $e^{\nu(r)} = C r^k$ (so that $\nu(r) = k \ln r + \ln C$). Here, the condition holds identically if  $k=2$.  This implies that \textit{every} radius $r$ admits circular null geodesics.
		If $k \neq 2$, the condition is never satisfied, and no circular null orbits exist.
		\item \textit{Linear exponent:} if $\nu(r) = kr$ (so the metric function grows exponentially $e^{\nu(r)} = e^{kr}$), the condition becomes $kr_C = 2$, yielding a unique solution  at $r_C = 2/k$.
\item \textit{Finslerian hairy black hole with flag curvature $K$:} Consider the hairy black hole solution discussed in Subsection~\ref{nekouee} which includes the flag curvature $K$ and a primary hair parameter equal to $\alpha \ell$:
\begin{equation*}
	e^{\nu(r)} = 1 - \frac{2M+\alpha\,\ell/K}{K\,r} + \alpha\,e^{-K r/M},
	\end{equation*}
(see \eqref{FHSchmetric} below).
The  condition $r_C \nu'(r_C) = 2$ is equivalent to  $r_C \frac{d}{dr}(e^{\nu(r)})|_{r=r_C }= 2 e^{\nu(r_C)}$, which in turn gives:
\begin{equation}\label{rC}
	\frac{3(2M + \alpha\,\ell/K)}{K\,r_C} - 2 = \alpha\,e^{-K r_C/M} \left( 2 + \frac{K\,r_C}{M} \right).
\end{equation}
This generalizes the Schwarzschild condition (recovered when $\alpha= 0$ and $K=1$). 
 The left-hand side in \eqref{rC} is a monotonically decreasing function from $+\infty$ (at $r=0$) to $-2$, while the right-hand side  is a bounded  function  that tends  monotonically to zero as $r\to +\infty$. Consequently, the two curves intersect at exactly one point $r_C$.
\end{itemize}

\end{remark}
\begin{remark}[Closed geodesics on $(S^2,\overline F)$]\label{lightring}
	Proposition~\ref{circulargeo}
	 does not require that the  projection of a null geodesic on $(S^2,\overline F)$ be closed. When $\overline F$ is given by 
	\eqref{Randersphere}, we effectively have  an  equatorial photon ring $\mathcal C:=\{r=r_C, \ \theta=\frac{\pi}{2}\}$, where $r_C$ is a solution of \eqref{photonsphere}.
	Away from the equator  the  projection on the photon sphere $r=r_C$ depends on the geodesic flow of $(S^2,\overline F)$ which  can be described in terms of the geodesics of the scaled round metric $h_K$ and the Killing vector field $W$, see \cite[Th. 2]{Robles2007}.
	Since  the reflection around the equator is an isometry   for  $\overline F_K$ in  \eqref{Randersphere},  $\mathcal C$ is  a light ring which is associated with   two opposite parametrized closed geodesics of $\overline F_K$. These geodesics have different $\overline F_K$--length and indeed they give rise to a Sagnac--type effect, see \S~\ref{sagnac}.
	\begin{proposition}[Equatorial closed geodesic for constant flag curvature Randers spheres]\label{equageos}
		Let $\overline F$ be a strongly convex Randers metric on $S^2$ with constant positive flag curvature $K$. Up to an $h_K$–isometry, choose spherical coordinates $(\theta,\phi)$ adapted to the symmetry axis so that $W=\varepsilon \partial_\phi$ (see Remark~\ref{rem:axial}). Then the equator $\{\theta=\tfrac{\pi}{2}\}$ is the support of two closed $\overline F$–geodesics with opposite parametrizations.
	\end{proposition}
	\begin{proof}
		Let $\mathcal R:(\theta,\phi)\mapsto(\pi-\theta,\phi)$ be the equatorial reflection. Since $\mathcal R$ is an $h_K$–isometry and $\mathcal R_*W=W$,  $\mathcal R$ is a $\overline F$–isometry by \cite[Lemma 1.2]{BaoRoblesShen2004}. Take $p$ on the equator and $v\in T_pS^2$ tangent to the equator and collinear to  $\partial_\phi$. Let $\gamma$ be the $\overline F$–geodesic with initial data $(p,v)$. Because $\mathcal R$ is an isometry, $\mathcal R\circ\gamma$ is a $\overline F$–geodesic with the same initial data (note that $d\mathcal R_p v=v$). By uniqueness, $\mathcal R\circ\gamma=\gamma$, so $\gamma$ lies in the fixed set of $\mathcal R$, i.e. the equator. Since the equator is a circle, $\gamma$ is closed. Replacing $v$ by $-v$ gives the opposite orientation, which is again tangent at $p$ and the same argument applies.
	\end{proof}
Propositions~\ref{circulargeo} and \ref{equageos} immediately give the following:
\begin{corollary}[Equatorial light ring]\label{2closed}
		For the static Lorentz--Finsler metric $L$ in \eqref{ansatz} with $\overline F$ as in \eqref{Randersphere}  any solution $r_C$ of $r_C\nu'(r_C)=2$ yields a photon ring $\mathcal C:=\{r=r_C,\ \theta=\tfrac{\pi}{2}\}$ consisting of the projection of two null circular future--pointing geodesics of $\overline F$ distinguished by the sign of $\dot\phi$.
	\end{corollary}
\end{remark}
\subsection{Axial Noether charge and critical impact parameter}\label{subsec:axial-charge}
For any Finsler metric $\overline F$ admitting a Killing vector field $X$, the Noether charge of $\overline L :=\frac 12 \overline F^2$ associated with $X$ is given by
\begin{equation*}
	\overline{\mathcal J}_X(x,y):= \frac{\partial \overline L}{\partial y^A}(x,y)X^A
	= \overline F(x,y)\frac{\partial \overline F}{\partial y^A}(x,y)X^A,
\end{equation*}
see, e.g., \cite{capcor}.
For a Randers metric $\overline F=\alpha+\beta$ with $\alpha$ and $\beta$ in \eqref{Randers}, one has
\begin{equation*}
	\frac{\partial \overline F}{\partial y^A}= \frac{1}{\alpha}\,a_{AB}y^B + b_A,
\end{equation*}
hence the constant of motion along a geodesic $\sigma$ is given by
\begin{equation*}
	\overline{\mathcal J}_X(\gamma, \dot\gamma)= \overline F(\sigma,\dot\sigma)\Big(\frac{a(\dot\sigma,X)}{\alpha(\sigma,\dot\sigma)}+\beta(X)\Big)
\end{equation*}
(as for the Noether charge $E$,  we will simply  denote such a geodesic--dependent constant with $\overline{\mathcal J}$).  For the axial symmetry $X=\partial_\phi$ of  Randers metric  $\overline F_K$ in \eqref{Randersphere} and for any  geodesic $\sigma(s)=(\theta(s),\phi(s))$, we then obtain:
\begin{equation}\label{Jbar}
	\overline{\mathcal J} = \overline F_K(\sigma,\dot\sigma)\Big(\frac{a_{\phi\phi}\dot\phi}{\alpha(\sigma,\dot\sigma)}+b_\phi\Big).
\end{equation}
where 
\begin{equation*}
	\alpha(\sigma,\dot\sigma)=\frac{1}{\sqrt{K}\,\lambda_K}\sqrt{\lambda_K\,\dot\theta^{\,2}+\sin^2\!\theta\,\dot\phi^{\,2}},\qquad a_{\phi\phi}=\dfrac{\sin^2\!\theta}{K\,\lambda_K^{2}},\qquad
	b_\phi=-\frac{\varepsilon\,\sin^2\!\theta}{K\,\lambda_K},
	\end{equation*}
and $\lambda_K$ 	is in \eqref{lambdak}.
On the photon sphere $r=r_C$ the null constraint  gives
\begin{equation*}
	e^{\nu(r_C)}\,\dot t^{\,2}=r_C^{\,2}\,\overline F^2_K(\theta,\phi;\dot\theta,\dot\phi).
\end{equation*}
By parametrizing $\sigma$ with the arc length 	$\overline F_K\big(\theta,\phi, \dot\theta,\dot \phi\big)\equiv 1$
and recalling that $E=\dot te^{\nu(r(s))}$, we then get 
\[  E=r_C\sqrt{e^{\nu(r_C)}}.\]
 For the full spacetime Lagrangian $L$ and the circular null geodesic $\gamma$ associated with  $\sigma$ (recall Proposition~\ref{circulargeo}) we have
\begin{equation}\label{J}
	\mathcal J =\frac{\partial (\tfrac12 L)}{\partial y^\phi}(\gamma,\dot\gamma)
	= -r^2_C\,\overline F_K(\sigma,\dot\sigma)\frac{\partial \overline F_K}{\partial y^\phi}(\sigma,\dot\sigma)
	= -r^2_C\overline{\mathcal J}.
	\end{equation}
Thus the {\em  impact parameter} $b:=\frac{|\mathcal J|}{E}$ at the photon sphere (the so-called {\em critical impact parameter})
is equal to 
\begin{equation}\label{bCany}
	b_C= \frac{r_C\sin^2\theta}{K\lambda_K\sqrt{e^{\nu(r_C)}}}\Big|\frac{\dot \phi}{\lambda_K\alpha(\sigma,\dot\sigma)}-\varepsilon\Big|.
\end{equation}
For the two equatorial geodesics at  $\theta=\tfrac{\pi}{2}$, recalling that $0\leq \varepsilon<\sqrt{K}$,  one then has 
\begin{equation}\label{bC}
	b_C^\pm=\frac{r_C}{K\lambda^*_K\sqrt{e^{\nu(r_C)}}}\Big|\sqrt{K}\mathrm{sgn}(\dot\phi)-\varepsilon\big|=
\frac{r_C}{K\lambda^*_K\sqrt{e^{\nu(r_C)}}}\Big(\sqrt{K}\mp\varepsilon\big),
\end{equation}
where $\lambda_K^*:=1-\varepsilon^2/K$, with  the minus (resp. plus) sign  corresponding to the  co-rotating (resp.    counter-rotating) geodesic, $\dot\phi>0$  (resp. $\dot\phi<0$). 
\begin{remark}
	\label{comparison}
	Lorentz--Finsler metrics of the type \eqref{ansatz} with $\overline F$ given by \eqref{Randersphere} have been considered recently in \cite{Nekouee2025}. We clarify two critical points regarding their analysis:
	\begin{enumerate}
		\item \textit{Symmetry assumptions:} the authors state that the spacetime is ``spherically symmetric'' and simplify the equations by focusing on the equatorial plane $\theta=\pi/2$. We emphasize that the Finslerian modification breaks the full $O(3)$ spherical symmetry to the axial group $SO(2)$ (see Remark~\ref{rem:axial}). Consequently, generic geodesics are non-planar. The restriction to the equator is a valid particular case only because it is a totally geodesic submanifold invariant under the discrete reflection $\theta \to \pi-\theta$, not because of spherical symmetry (see Proposition~\ref{equageos}).
		\item \textit{Reversibility and impact parameter:} crucially, the authors do not account for the fact that $\overline F$ is not reversible. 
		Their expression for angular momentum on the plane $\theta=\pi/2$  is $\mathcal{J} = r^2 (1+\epsilon)^{-2} \dot{\phi}$ (see  Eq. (82), where the authors assume $K=1$). This formula is correct \emph{only} for  co-rotating geodesics ($\dot{\phi}>0$).
		For  counter-rotating ones ($\dot{\phi}<0$), recalling that $\dot{\phi} = -(1-\epsilon)$ in the arc-length parametrization of $F_K$, the correct relation derived from \eqref{Jbar}--\eqref{J} is:
		\[
		\mathcal{J} = r^2 (1-\epsilon)^{-2} \dot{\phi}.
		\]
		By using a single relation, reference \cite{Nekouee2025} misses in Eq. (88) the splitting of the critical impact parameter into two distinct values $b_C^\pm$, see  \eqref{bC}.
	\end{enumerate}
\end{remark}

\subsection{Orbit equation}
While the general Finslerian spacetime lacks full spherical symmetry, the existence of the reflection isometry  ensures that the equatorial plane $\theta = \pi/2$ is a totally geodesic submanifold. We restrict our dynamical analysis to this invariant plane to derive an orbit equation.

On the equator $\theta=\pi/2$, \eqref{Jbar} and \eqref{J} (with $r=r(\tau)$ instead of the constant value $r_C$) give:  
\begin{equation}\label{JH}
	\mathcal{J} = -r^2 (\dot{\phi} \mathcal{H}) \mathcal{H} = -r^2 \dot{\phi} \mathcal{H}^2,
\end{equation}
where
\[\mathcal{H}:= \frac{\text{sgn}(\dot{\phi})\sqrt{K} - \epsilon}{K \lambda_K^*}.\]
Solving for $\dot{\phi}$ and substituting it back into the expression for $\overline{F}_K$ allows us to find the squared angular metric in terms of $\mathcal{J}$:
\begin{equation}
	\label{F_squared}
	\overline{F}_K^2 = (\dot{\phi} \mathcal{H})^2 =  \frac{\mathcal{J}^2}{r^4} \frac{1}{\mathcal{H}^2} = \frac{\mathcal{J}^2}{r^4} \left( \frac{K \lambda_K^*}{\sqrt{K} \mp \epsilon} \right)^2.
\end{equation}
Using $E = e^{\nu(r)}\dot{t}$, the null condition  and  \eqref{F_squared}, we then get
\begin{equation*}
	e^{\vartheta(r)} \dot{r}^2 = E^2 e^{-\nu(r)} - r^2 \overline{F}_K^2 = E^2 e^{-\nu(r)} - \frac{\mathcal{J}^2}{r^2} \left( \frac{K \lambda_K^*}{\sqrt{K} \mp \epsilon} \right)^2,
\end{equation*}
thus by  $\dot{r} = (dr/d\phi)\dot{\phi}$ and \eqref{JH}, we obtain
\begin{equation*}
	\left(\frac{dr}{d\phi}\right)^2 = \frac{r^4\mathcal H^4}{e^{\nu(r)+\vartheta(r)}} \left[ \frac{1}{b^2} - \frac{e^{\nu(r)}}{r^2} \left( \frac{K \lambda_K^*}{\sqrt{K} \mp \epsilon} \right)^2 \right].
\end{equation*}
The turning points of the orbit must satisfy $dr/d\phi = 0$. Setting the term in the square brackets to zero, we find the impact parameter required to reach a turning point at $r_{tp}$ must satisfy
\begin{equation*}
	b = \frac{r_{tp}}{K \lambda_K^* \sqrt{e^{\nu(r_{tp})}}} \left( \sqrt{K} \mp \epsilon\right).
\end{equation*}
in accordance to  the critical impact parameter derived in \eqref{bC}.

\subsection{A Finslerian  Sagnac--type effect}\label{sagnac}
The time delay between different images in a lensing system also receives Finslerian corrections. Since co-rotating and counter-rotating paths have different effective metrics due to non-reversibility of $L$, the travel time for light differs depending on the orientation of the path around the lensing mass leading to a gravitational Sagnac--type effect. Let us describe it in detail.

For null future--pointing curves projecting on the photon sphere $r=r_C$, we have
\begin{equation*}
	\dot t=\frac{r_C}{\sqrt{e^{\nu(r_C)}}}\,\overline F_K(\theta,\phi, \dot \theta,\dot \phi),
\end{equation*} 
so that along the two equatorial closed geodesics of $\overline F_K$, $\sigma^\pm$,  in Proposition~\ref{equageos} at the photon sphere $r_C$, the lapse of flight for the $t$ coordinate  is
\begin{equation*}
	T_\pm=\frac{r_C}{\sqrt{e^{\nu(r_C)}}}\int_0^{2\pi}\!\Big(\alpha(\dot\sigma^\pm )+ \beta(\dot\sigma^\pm)\Big)\,ds.
	\end{equation*} 
	Hence the \emph{Sagnac--type} time delay between co- and counter-rotating equatorial loops at $r=r_C$ as measured by the observer $\partial_t/\sqrt{e^{\nu(r_C)}}$ is 
\begin{equation*}
	|\Delta \tau|:=\frac{1}{\sqrt{e^{\nu(r_C)}}}|T_+-T_-|=\frac{2r_C}{e^{\nu(r_C)}}\Big|\int_{\sigma^+}\!\beta\Big|
	=\frac{4\pi r_C}{e^{\nu(r_C)}}\frac{\varepsilon}{K\lambda_K^*}= \frac{4\pi r_C}{e^{\nu(r_C)}} \frac{\varepsilon}{K-\varepsilon^2}.
\end{equation*}

Note that, as shown in~\cite{CaponioMasiello2022}, a Sagnac--type delay as above can also arise in Lorentzian stationary spacetimes. This is not surprising as the future (resp. the past) null cones of any stationary spacetime can be locally seen as the future (resp. past) ones of a static Finsler spacetime \cite[\S 7]{CapSta18}; moreover these two spacetimes share (up to reparametrization) the same future--pointing (resp. past--pointing) null geodesics \cite[\S 6 and Appendix B]{CapSta18}.  To assess whether the static Lorentz--Finsler metric \eqref{ansatz} produces effects distinguishable from a stationary model, we consider in the next section  a purely timelike, clock-based test in two isocausal  models.

\section{Distinguishing static Finsler from stationary Lorentzian spacetimes under shared optical data}
\label{appendix}

In this section we present a way  to discriminate between a stationary Lorentzian metric $g$ and a static Lorentz--Finsler metric $L$ of the form \eqref{ansatz}. We perform the comparison in the most stringent setting, namely when the angular Randers metric $\overline F$ is exactly the one obtained from $g$ by the stationary–to–Randers correspondence \cite{CJS}. Without loss of generality we choose the time unit so that the lapse equals $1$ on $S_{r_0}:=\{r_0\}\times S^2$ in both models, where $r_0$ is a solution of $r\nu'(r)=2$. This can be done by a constant rescaling of $t$ and does not affect proper–time measurements along $S_{r_0}$. For simplicity, we then take the lapse to be identically $1$ in both models.

\subsubsection*{Stationary spacetime:}
\begin{itemize}
	\item the spacetime is $M=\mathbb R_t\times\mathbb R_r\times S^2$, and we set $\Sigma:=\mathbb R_r\times S^2$;
	\item let $h$ be a Riemannian metric on $S^2$ and $\omega$ a $1$–form on $S^2$ with $\|\omega\|_{h}<1$; we extend $\omega$ to $\R_r\times S^2$ as $\hat\omega(r,\cdot):=r\,\omega(\cdot)$, so $\hat\omega(\partial_r)=0$;
	\item  each slice $\{t=\bar t\}$ is spacelike and the  Riemannian metric induced by $g$ is  
	\[g_0:=e^{\nu(r)}dr^2\oplus r^2 h-\hat\omega\otimes\hat\omega=e^{\nu(r)}dr^2\oplus r^2( h-\omega\otimes\omega);\]
	\item the stationary spacetime metric is
	\begin{equation}\label{eq:stationary-standard}
		g=(dt-\hat\omega)^2 - e^{\nu(r)}dr^2 - r^2 h,
	\end{equation}
	(with the shorthand $(dt -\hat\omega)^2:=(dt-\hat\omega)\otimes(dt -\hat\omega)$ and the abuse of notation   concerning the pullback of tensors to $M$ that are denoted as the original ones).
\end{itemize}
\subsubsection*{Static Lorentz–Finsler spacetime:}
\begin{itemize}
	\item the spacetime is $M=\mathbb R_t\times\mathbb R_r\times S^2$.
	\item the angular Finsler metric on $S^2$ is of Randers type
	\begin{equation*}
		\overline F(v)=\alpha(v)+\omega(v),\qquad \alpha=\sqrt{h(v,v)},\qquad v\in TS^2
	\end{equation*}
	(here we avoid writing explicitly the base point $x\in S^2$  carried by $v$)
	with the same $h$ and $\omega$ as above.
	\item the Lorentz–Finsler metric is
	\[
	L(y^t,y^r,v)=(y^t)^2 - e^{\nu(r)}(y^r)^2 - r^2 \overline F^2(v) ,
	\]
	for each $(y^t,y^r,v)\in \R\times\R\times TS^2$.
\end{itemize}
For the stationary metric \eqref{eq:stationary-standard}, the {\em Fermat (optical) metric} on the shell $S_{r_0}$ is
\[
F_g\big|_{TS_{r_0}}(v)=\sqrt{r_0^2 h(v,v)}+r_0\omega(v) = r_0\big(\sqrt{h(v,v)}+\omega(v)\big)=r_0\overline F(v),
\]
for each $v\in TS^2$. For $L$, one likewise has $F_L|_{TS_{r_0}}(v)=r_0\,\overline F(v)$ (see \cite{CJS,CapSta18}). Hence both metrics $g$ and $L$  share the same future–pointing null geodesics, with constant radial component $r\equiv r_0$, provided  they exist, i.e. when $r_0\nu'(r_0)=2$ is satisfied  (cf.\ \cite[Prop.~4.1]{CJS} and \cite[Prop.~B.1]{CapSta18} and Proposition~\ref{circulargeo}). In particular, the time delay in the Sagnac--type effect described in \S\ref{sagnac} is the same in both geometries.
In order to distinguish the two geometries  we consider the following test involving some future--pointing timelike curves rather than null geodesics.

Let us fix a smooth closed loop $\sigma:[0,L]\to S_{r_0}$, parametrized by the $h$–arclength $s$ on $S^2$, so that $T:=\dot\sigma$ obeys $h(T,T)\equiv1$. Let us set $\omega(T)(s):=\omega_{\sigma(s)}\big(T(s)\big)$. 
A \emph{time–stretch profile} is a function $\mathfrak b:[0,L]\to(1,\infty)$ prescribing the $t$–speed along $\sigma$.

Along $\sigma$ we have  $dr\equiv0$, and
so $\frac{dt}{ds}=\mathfrak b$ must satisfy the following inequalities if we want that $t=t(s)$ is the component of a  future--pointing timelike curve for both orientation of $\sigma$:
\begin{align*}
	\textit{stationary:}\quad
	& \big(\mathfrak b\mp r_0\omega(T)\big)^2 - r_0^2>0\\
	\textit{static Lorentz--Finsler:}\quad
	&\mathfrak  b^2 - r_0^2\big(1\pm\omega(T)\big)^2>0
\end{align*}
(we stress that the second right-hand sides in both inequalities above are evaluated at $L-s$ and not $s$ when we consider the opposite parametrization of $\sigma$). Hence, recalling that $\|\omega\|_h<1$,  the condition 
\begin{equation*}
	\qquad \mathfrak b(s) > r_{0}\big(1+|\omega(T)(s)|\big),\qquad \text{for all }s\in[0,L].\qquad
\end{equation*}
guarantees that the curves $\gamma^\pm(s)=\big(\int_0^s \mathfrak b(u)du, r_0, \sigma^\pm(s)\big)$ are  future–pointing and timelike in both models and both orientations,  $\sigma^{+}$ and $\sigma^-$, of  $\sigma$.
Hence the proper times of these curves, depending on the profile $\mathfrak b$, are then
\begin{align}
	\tau_{g}^\pm(\mathfrak b)
	&=\int_0^L \sqrt{\big(\mathfrak b(s)\mp r_0\,\omega(T)(s)\big)^2-r_0^2}\;ds, \label{taug}\\
	\tau_{L}^{\pm}(\mathfrak b)
	&=\int_0^L \sqrt{\mathfrak b(s)^2-r_0^2\big(1\pm \omega(T)(s)\big)^2\,}\;ds. \label{tauL}
\end{align}
where  we have used the change of variable $s\mapsto L-s$ for $\tau_g^-$ and $\tau_L^-$.

Let us write, for brevity, $\omega=\omega(T)(s)$ and $-\omega=\omega(-T)(L-s)$

\subsubsection*{Stationary:}
From \eqref{taug} and $\mathfrak b\gg 1$ we factor the large terms $\mathfrak b\mp r_0\omega$:
\[
\sqrt{(\mathfrak b\mp r_0\omega)^2-r_0^2}=(\mathfrak b\mp r_0\omega)\,
\sqrt{1-\Big(\tfrac{r_0}{\mathfrak b\mp r_0\omega\,}\Big)^{\!2}},
\]
so that
\[
\sqrt{(\mathfrak b\mp r_0\omega)^2-r_0^2}
=(\mathfrak b\mp r_0\omega)
-\frac{r_0^2}{2(\mathfrak b\mp r_0\omega)}
+O\!\Big(\frac{r_0^4}{(\mathfrak b\mp r_0\omega)^3}\Big).
\]
Hence
\bal
\tau_{g}^{+}(\mathfrak b)-\tau_{g}^{-}(\mathfrak b)
&=\!\int_0^L\!\!\Big[-2r_0\omega
-\frac{r_0^2}{2}\Big(\frac{1}{\mathfrak b-r_0\omega}-\frac{1}{\mathfrak b+r_0\omega}\Big)
+O\!\Big(\frac{r_0^4}{(\mathfrak b\mp r_0\omega)^3}\Big) \Big]ds\nonumber\\
&=\!\int_0^L\!\!\Big[-2r_0\omega
-\frac{r_0^3\omega}{\mathfrak b^2-r_0^2\omega^2}
+O\!\Big(\frac{r_0^4}{(\mathfrak b\mp r_0\omega)^3}\Big) \Big]ds.\label{deltataug}
\eal

\subsubsection*{Static Lorentz--Finsler:}
From \eqref{tauL} and  $\mathfrak b\gg1$ we have:
\[
\sqrt{\mathfrak b^2-r_0^2(1\pm\omega)^2}
= \mathfrak b\,\sqrt{\,1-\Big(\tfrac{r_0}{\mathfrak b}\Big)^{\!2}(1\pm\omega)^2}\ .
\]
Then
\[
\sqrt{\mathfrak b^2-r_0^2(1\pm\omega)^2}
= \mathfrak b - \frac{r_0^2(1\pm\omega)^2}{2\mathfrak b}
+ O\!\Big(\frac{r_0^4}{\mathfrak b^3}\Big).
\]
Hence
\bal
\tau_{L}^{+}(\mathfrak b)-\tau_{L}^{-}(\mathfrak b)
&=\!\int_0^L\!\!\Big[
-\frac{r_0^2}{2\mathfrak b}\big((1+\omega)^2-(1-\omega)^2\big)
+O\!\Big(\frac{r_0^4}{\mathfrak b^3}\Big)\Big]ds\nonumber\\
&=\!\int_0^L\!\!\Big[
-\frac{2r_0^2\omega}{\mathfrak b}
+O\!\Big(\frac{r_0^4}{\mathfrak b^3}\Big)\Big]ds\label{deltatauL}
\eal
Let us define for $\mathfrak b\gg 1$, 
\[\mu:=\min_{s\in [0,L]}\big(\mathfrak b(s)-r_0|\omega(T)(s)|\big)>0.\]
Hence we get the following estimates in both models.

\subsubsection*{Stationary:} From \eqref{deltataug} 
\[
\tau_{g}^{+}(\mathfrak b)-\tau_{g}^{-}(\mathfrak b)
= -\,2r_0\!\int_0^L \omega(T)\,ds \;+\; O\Big(\frac{1}{\mu^2}\big).
\]
\subsubsection*{Static Lorentz--Finsler:} As $\mu\leq\min_{s\in[0,L]}\mathfrak b(s)$, from \eqref{deltatauL} 
\[
\tau_{L}^{+}(b)-\tau_{L}^{-}(\mathfrak b)
= -\,2r_0^2\!\int_0^L\!\frac{\omega(T)}{\mathfrak b(s)}\,ds \;+\; O\Big(\frac{1}{\mu^3}\Big)\ .
\]
In particular, in the stationary spacetime $(M,g)$ the leading term is  independent of $\mathfrak b$, whereas in the static Finsler spacetime $(M,L)$ the  leading term depends on the chosen profile.

Thus, on the sphere  $S_{r_0}$ the two models share the same optical data $F|_{TS_{r_0}}=r_0\overline F$ and the same  null geodesics with constant $r$--component, hence the Sagnac--type delays are identical. Nevertheless,  the  difference of proper time intervals  of some  future-pointing timelike curves projecting on the same  loop $\sigma$ traversed in both   orientations, distinguishes  the models.

\section{Finslerian field equations and their simplification}\label{FFE}
Black holes in static Finsler spacetimes have recently been studied in \cite{Nekouee2025}, where the authors claim to  obtain novel solutions via the extended gravitational decoupling (EGD) method \cite{Ovalle2019}, which they term   \emph{Finslerian hairy black holes}. In what follows we examine their work and point out  an attribution issue.

\subsection{The vacuum equations}
One of the remarkable features of static Finsler spacetimes with constant flag curvature angular sector is that the field equations simplify dramatically. Li and Chang \cite{LiChang2014} considered the Finslerian vacuum field equation proposed by Rutz \cite{Rutz}, which is a scalar equation on $\mathcal A$ given as $\text{Ric} = 0$, where Ric is the Ricci scalar of the Finsler metric $L$ (see \cite[\S 7.6]{Bao2000}).\footnote{Note that the Ricci scalar can be computed in $\mathcal  A$ starting with the positively homogeneous of degree $2$ function $L$, see e.g. \cite{CapMas20}. If  $\mathrm{Ric}=0$ on $\mathcal A$ then we can say that it is $0$ on $TM$, recall Remark~\ref{regularity}.} The constant flag curvature condition implies that the Finslerian Ricci scalar of $\overline F$ is also constant, $\overline{\text{Ric}} = K$ (see \cite{Bao2000}, Exercise 7.6.2). 
The nice observation in \cite{LiChang2014} is that when the angular metric $\overline{F}$ has constant Ricci scalar  $\overline{\text{Ric}} =K$, the equation $\mathrm{Ric}=0$ reduces to three ordinary differential equations in the radial coordinate alone and the anisotropy, relegated to the angular sector, disappears completely.   Actually, the fact that the angular part does not enter the field equations in the spherical symmetric ansatz is well-known in the Lorentzian framework (see \cite{Poisson}, p. 266). In conclusion, the equations obtained are the same as in the static, spherical symmetric Lorentzian setting with the only difference given by the curvature constant $K$ that can be different from $1$. Thus the time and radial components of the solution after fixing some constants of integration take the form \cite[Eqs. (61)--(62)]{LiChang2014}:
\begin{equation}\label{FSc}
	e^\nu = 1 - \frac{2M}{K r}, \quad e^{-\vartheta} = K - \frac{2M}{r},
\end{equation} 
where $M$ is a mass parameter. 

\subsection{The Finslerian hairy black hole solutions}\label{nekouee} 
In the presence of matter, the authors in \cite{Nekouee2025} follow   \cite{LiChang2014}, where
the Akbar--Zadeh Ricci tensor and scalar \cite{AkbarZadeh1988} are considered: 
\[
\mathrm{Ric}_{ij}:=\frac{1}{2}\,\frac{\partial^{2}}{\partial y^{i}\partial y^{j}}\big(L\mathrm{Ric}\big),
\qquad
S:=g^{ij}\mathrm{Ric}_{ij},
\]
from which they build  the Finslerian Einstein tensor
$G_{\mu\nu}=\mathrm{Ric}_{ij}-\frac{1}{2}g_{ij}S$
and then the field equations on $\mathcal A$:
 \begin{equation}\label{eq:Finsler-EFE}
 	G^\mu_\nu \;:=\; \mathrm{Ric}^\mu_\nu - \tfrac{1}{2}\,\delta^\mu_\nu\,S \;=2\operatorname{area}_{\mathcal F_K}(S^2)T^\mu_\nu,
 \end{equation}
obtained by raising indexes using the inverse $g^{ij}$ of the fundamental tensor $g$ in \eqref{fundtens}; 
here $\operatorname{area}_{\overline F_K}(S^2)$
is the area of the Randers sphere $(S^2, \overline F_K)$ with respect to a fixed  volume form (e.g. the 
Busemann--Hausdorff one, see \cite[Example 2.2.2]{Shen01}) that the authors of \cite{LiChang2014,Nekouee2025} denote with $4\pi_{\overline F}$.
We emphasize that \eqref{eq:Finsler-EFE} is introduced by \emph{analogy} with the Einstein field equations and it is not obtained from a variational principle.
It is important to say that action based formulations of Finsler gravity exist in literature at least when $T^\mu_\nu=0$, see  \cite{PfeiferWohlfarth2012,HohmannPfeiferVoicu2019,GarciaParradoMinguzzi2022GRG,JavaloyesSanchezVillasenor}.

The authors of \cite{Nekouee2025} start from equation \eqref{eq:Finsler-EFE} considering a   so-called {\em anisotropic fluid}: the  stress--energy tensor  is diagonal with components depending only on $r$ and the  tangential and radial pressure are different, i.e. $T_{rr}\neq T_{\theta\theta}=T_{\phi\phi}$.

Again the constant curvature assumption on $\overline F$ reduces essentially the field equations to the standard Einstein field equations in the static spherically  symmetric setting (see \cite[Eq. (82.2)]{Tolman}):
\begin{equation*}
	G^t_t=2\operatorname{area}_{\mathcal F_K}(S^2)T^t_t,\qquad G^r_r=2\operatorname{area}_{\mathcal F_K}(S^2)T_r^r,\qquad G^\theta_\theta=2\operatorname{area}_{\mathcal F_K}(S^2)T^\theta_\theta,
\end{equation*}
where
\begin{align}
	G^t_{t} &= e^{-\vartheta}\!\left(\frac{\vartheta'}{r}-\frac{1}{r^2}\right)+\frac{K}{r^2}, \nonumber\\
	G^r_{r} &= e^{-\vartheta}\!\left(-\frac{\nu'}{r}-\frac{1}{r^2}\right)+\frac{K}{r^2}, \nonumber\\
	G^\theta_{\theta} &= G^\phi_{\phi} \;=\; \frac{e^{-\vartheta}}{4}\!\left(2\nu''+(\nu')^2-\nu'\vartheta' + \frac{2}{r}(\nu'-\vartheta')\right). \label{eq:Gang}
\end{align}
Equations \eqref{eq:Finsler-EFE}---\eqref{eq:Gang} reproduce the vacuum solution quoted above when $T^\mu_\nu\equiv0$ \cite{LiChang2014}. To incorporate hair, the authors  follow the (extended) gravitational decoupling (EGD) prescription \cite{Ovalle2019}: they 
deform the vacuum solution \eqref{FSc} as follows:
\begin{equation*}
	e^{-\vartheta(r)}=\mu(r)+\alpha\,g(r),\qquad 
	e^{\nu(r)}=\frac{\mu(r)}{K}+\alpha\,h(r),\qquad 
	\mu(r):=K-\frac{2M}{r},
\end{equation*}
so that the empty solution is recovered at the parameter $\alpha=0$.  It is then not surprising that the solution in \cite{Nekouee2025} bears a striking resemblance to the hairy black hole solution obtained by Ovalle et al. \cite{Ovalle2021} using the gravitational decoupling method in the same static, spherical  framework but with the standard round sphere as the angular sector.
In fact, the solution in \cite{Nekouee2025} is (see \cite[Eq. (67)]{Nekouee2025}):
\begin{equation}\label{FHSchmetric}
	e^{\nu(r)}=e^{-\vartheta(r)}
	= 1-\frac{2M+\alpha\,\ell/K}{K\,r}+\alpha\,e^{-K r/M},
\end{equation}
while  Ovalle et al.'s  solution takes the form (see \cite[Eq. (69)]{Ovalle2021}):
\begin{equation}\label{ovalle}
	e^{\nu(r)} = e^{-\vartheta(r)} = 1 - \frac{2\mathcal{M}}{r} + \alpha e^{-r/(\mathcal{M}-\alpha\ell/2)},
\end{equation} 
where $\mathcal{M} = M + \alpha\ell/2$.  In both cases the constant $\alpha\ell$ is the so-called {\em primary hair}.  
When we account for the different notation and the rescaling by the flag curvature $K$, these solutions are mathematically identical off the angular sector. More specifically, setting $K = 1$, \eqref{FHSchmetric} reproduces \eqref{ovalle} exactly. 

In conclusion, the derivation in Nekouee et al. follows precisely the same extended gravitational decoupling method, applying the same techniques to arrive at the same result (comparing pages 8--12 in \cite{Nekouee2025} with pages 3--5 in \cite{Ovalle2021} reveals that the calculations are essentially identical).

Despite this clear correspondence, Nekouee et al.\ do not cite \cite{Ovalle2021} by Ovalle et al., even though they cite other works by Ovalle and collaborators on the gravitational decoupling method.

\section{Conclusions}\label{end}
We have presented an analysis of static Lorentz--Finsler metrics in spherical coordinates  where the ``Finslerianity'' is confined  to the sphere $S^2$ and we have observed  that the condition for circular null geodesics (photon sphere) in these class of static Finsler spacetimes  takes the same form as in general relativity: 
\[r\nu'(r) = 2.\] 
This result holds independently of the specific angular Finsler structure, provided this Finsler structure  does not depend on the radial coordinate. 

When the Finsler metric $\overline F$ on $S^2$ is a Randers metric with constant flag curvature, 
we have  evaluated the critical impact parameter (defined as the ratio  $|\mathcal J|/E$ between the constants of motion $\mathcal J$ and $E$ at the photon sphere $r=r_C$) for any circular null future--pointing geodesic, see \eqref{bCany}. We have obtained the  orbit equation on the totally geodesic plane $\theta=\pi/2$ and  the value of the  impact parameter at a turning point. We also have quantified the asymmetry of $\overline F$ through a Sagnac--type effect. Since the same effect can be measured in a stationary Lorentzian spacetime, we have showed in Section~\ref{appendix} that this class of spacetimes can be distinguished from static Lorent-Finsler one by a clock-based test.
% showing that the time delay between co- and counter-rotating equatorial loops is given by 
%\[|\Delta \tau| = \frac{4\pi r_C}{e^{\nu(r_C)}} \frac{\varepsilon}{K-\varepsilon^2},\] 
%for a  Randers sphere with positive constant flag curvature $K$ and wind $\varepsilon\partial_\phi$, $0\leq \varepsilon<\sqrt{K}$. 

Finally, we have addressed the relationship with the recent results in \cite{Nekouee2025}. While this reference investigates the thermodynamics of the hairy black hole model discussed therein, its optical analysis is incomplete as it neglects the intrinsic non-reversibility of the Finsler metric. Consequently, it fails to capture the  difference  between the co-rotating and counter-rotating critical impact parameters.
Furthermore, regarding the solution itself, we clarify that the hairy black hole spacetime presented in \cite{Nekouee2025} is not a novel derivation; the metric coefficients  in the $(t,r)$--sector are mathematically identical to those obtained by Ovalle et al. \cite{Ovalle2021}, differing only by a constant rescaling involving the flag curvature parameter $K$.

\section*{Acknowledgments}
	\noindent  We sincerely thank the  referees for their   valuable comments. 
	
	\noindent This work  is  partially supported by  PRIN 2022 PNRR {``P2022YFAJH Linear and Nonlinear PDE's: New directions and Application''},  by  MUR under the Programme ``Department of Excellence'' Legge 232/2016  (Grant No. CUP - D93C23000100001) and by  GNAMPA INdAM - Italian National Institute of High Mathematics.

\end{document}